\documentclass[journal]{IEEEtran}
\usepackage{amsfonts}
\usepackage{amsmath}
\usepackage{amsthm}
\usepackage{amssymb}
\usepackage{graphicx}
\usepackage[T1]{fontenc}
\usepackage{supertabular}
\usepackage{longtable}
\usepackage[usenames,dvipsnames]{color}
\usepackage{bbm}
\usepackage{fancyhdr}
\newcommand{\p}{\partial}
\usepackage{breqn}
\usepackage{fixltx2e}
\usepackage{capt-of}
\usepackage{mdframed}
\usepackage{tikz}
\usetikzlibrary{matrix}
\usepackage{endnotes}
\usepackage{soul}

\newtheorem{proposition}{Proposition}

\newcommand{\mathsym}[1]{}
\newcommand{\unicode}[1]{}

\begin{document}
\title{ Election Predictions as Martingales: An Arbitrage Approach}
\author{
 Nassim Nicholas Taleb\\
Tandon School of Engineering, New York University 
3rd Version, October 2017\\
Forthcoming, \textit{Quantitative Finance} }
\maketitle
\thispagestyle{fancy}

\section{Introduction}

A standard result in quantitative finance is that when the volatility of the underlying security increases, arbitrage pressures push the corresponding binary option to trade closer to 50\%, and become less variable over the remaining time to expiration. Counterintuitively, the higher the uncertainty of the underlying security, the lower the volatility of the binary option. This effect should hold in all domains where a binary price is produced -- yet we observe severe violations of these principles in many areas where binary forecasts are made, in particular those concerning the U.S. presidential election of 2016. We observe stark errors among political scientists and forecasters, for instance with 1) assessors giving the candidate D. Trump between 0.1\% and 3\% chances of success , 2) jumps in the revisions of forecasts from 48\% to 15\%, both made while invoking uncertainty.

Conventionally, the quality of election forecasting has been assessed statically by De Finetti's method, which consists in minimizing the Brier score, a metric of divergence from the final outcome (the standard for tracking the accuracy of probability assessors across domains, from elections to weather). No intertemporal evaluations of changes in estimates appear to have been imposed outside the quantitative finance practice and literature. Yet De Finetti's own principle is that a probability should be treated like a two-way "choice" price, 
which is thus violated by conventional practice.
 \begin{figure}[h!tb]
\begin{center}
	\includegraphics[width=0.9 \linewidth]{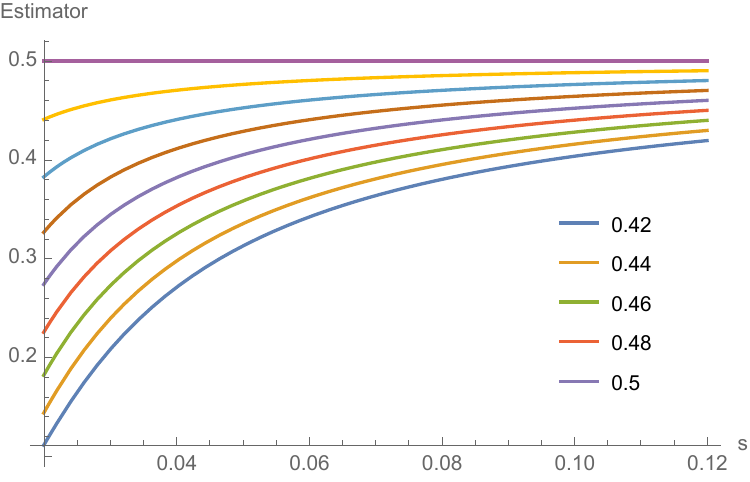}
	\caption{Election arbitrage "estimation" (i.e., valuation) at different expected proportional votes $Y\in [0,1]$, with $s$ the expected volatility of $Y$ between present and election results. We can see that under higher uncertainty, the estimation of the result gets closer to 0.5, and becomes insensitive to estimated electoral margin.}\label{graphmain}
	\end{center}
\end{figure}

\begin{figure}[h!tb]
\begin{center}
	\includegraphics[width=0.6 \linewidth]{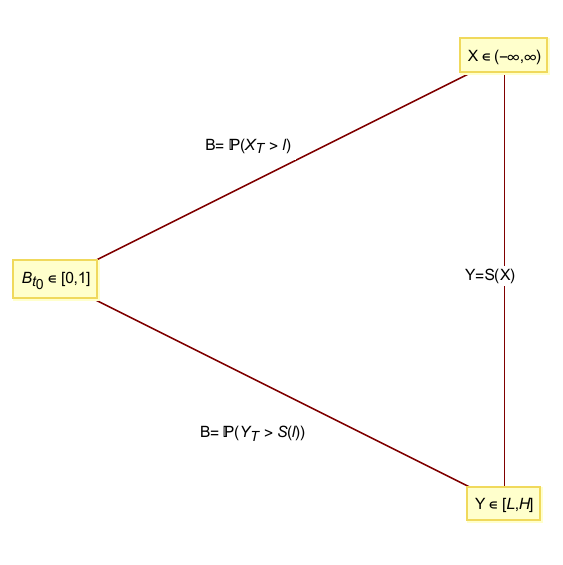}
	\caption{$X$ is an open non observable random variable (a shadow variable of sorts) on $\mathbb{R}$, $Y$, its mapping into "votes" or "electoral votes" via a sigmoidal function $S(.)$, which maps one-to-one, and the binary as the expected value of either using the proper corresponding distribution.}\label{scheme}
	\end{center}
\end{figure}

In this paper we take a dynamic, continuous-time approach based on the principles of quantitative finance and argue that a probabilistic estimate of an election outcome by a given "assessor" needs be treated like a tradable price, that is, as a binary option value subjected to arbitrage boundaries (particularly since binary options are actually used in betting markets). Future revised estimates need to be compatible with martingale pricing, otherwise intertemporal arbitrage is created, by "buying" and "selling" from the assessor.

A mathematical complication arises as we move to continuous time and apply the standard martingale approach: namely that as a probability forecast, the underlying security lives in $[0,1]$. Our approach is to create a dual (or "shadow") martingale process $Y$, in an interval $[L,H]$ from an arithmetic Brownian motion, $X$ in $(-\infty, \infty)$ and price elections accordingly. The dual process $Y$ can for example represent the numerical votes needed for success. A complication is that, because of the transformation from $X$ to $Y$, if $Y$ is a martingale, $X$ cannot be a martingale (and vice-versa). 

The process for $Y$ allows us to build an arbitrage relationship between the volatility of a probability estimate and that of the underlying variable, e.g. the vote number. Thus we are able to show that when there is a high uncertainty about the final outcome, 1) indeed, the arbitrage value of the forecast (as a binary option) gets closer to 50\% and 2) the estimate should not undergo large changes even if polls or other bases show significant variations.\footnote{A central property of our model is that it prevents $B(.)$ from varying more than the estimated $Y$: in a two candidate contest, it will be capped (floored) at $Y$ if lower (higher) than .5. In practice, we can observe probabilities of winning of 98\% vs. 02\% from a narrower spread of estimated votes of 47\% vs. 53\%; our approach prevents, under high uncertainty, the probabilities from diverging away from the estimated votes. But it remains conservative enough to not give a higher proportion.}

The pricing links are between 1) the binary option value (that is, the forecast probability), 2) the estimation of $Y$ and 3) the volatility of the estimation of $Y$ over the remaining time to expiration (see Figures \ref{graphmain} and \ref{scheme} ). 


\subsection{Main results}
For convenience, we start with our notation.
\subsubsection*{Notation} 

\begin{center}
\begin{tabular}{l p{0.8 \linewidth}}
$Y_0$ & the observed estimated proportion of votes expressed in $[0,1]$ at time $t_0$. These can be either popular or electoral votes, so long as one treats them with consistency.\\
$T$ &period when the irrevocable final election outcome $Y_T$ is revealed, or expiration.\\

$t_0$ &present evaluation period, hence $T-t_0$ is the time until the final election, expressed in years.\\

$s$ &annualized volatility of $Y$, or uncertainty attending outcomes for $Y$ in the remaining time until expiration. We assume $s$ is constant without any loss of generality --but it could be time dependent.\\

$B(.)$ &"forecast probability", or estimated continuous-time arbitrage evaluation of the election results, establishing arbitrage bounds between $B(.)$,  $Y_0$ and the volatility $s$.

\end{tabular}
\end{center}

\subsubsection*{Main results}

\begin{equation}
B(Y_0,\sigma,t_0,T)=\frac{1}{2} \text{erfc}\left(\frac{l-\text{erf}^{-1}(2 Y_0-1) e^{\sigma ^2 ( T-t_0)}}{\sqrt{e^{2 \sigma ^2 (T-t_0)}-1}}\right)\label{maineq},
\end{equation}
where 
\begin{equation}
\sigma \approx \frac{\sqrt{\log \left(2 \pi   s^2 e^{2 \text{erf}^{-1}(2 Y_0-1)^2}+1\right)}}{\sqrt{2} \sqrt{T-t_0}}	\label{eqvar},
\end{equation}
$l$ is the threshold needed (defaults to .5), and erfc(.) is the standard complementary error function, 1-erf(.), with $
\text{erf} (z)=\frac{2}{\sqrt{\pi }}\int _0^z  e^{-t^2} dt$. \qed

We find it appropriate here to answer the usual comment by statisticians and people operating outside of mathematical finance: "why not simply use a Beta-style distribution for $Y$?". The answer is that 1) the main purpose of the paper is establishing (arbitrage-free) time consistency in binary forecasts, and 2) we are not aware of a continuous time stochastic process that accommodates a beta distribution or a similarly bounded conventional one. 

\subsection{Organization}
The remaining parts of the paper are organized as follows. First, we show the process for $Y$ and the needed transformations from a specific Brownian motion. Second, we derive the arbitrage relationship used to obtain equation \eqref{maineq}.  Finally, we discuss De Finetti's approach and show how a martingale valuation relates to minimizing the conventional standard in the forecasting industry, namely the Brier Score.


 \subsubsection*{A comment on absence of closed form solutions for $\sigma$}
 We note that for $Y$ we lack a closed form solution for the integral reflecting the total variation: $\int_{t_0}^T \frac{\sigma  }{\sqrt{\pi }}e^{-\text{erf}^{-1}(2 y_s-1)^2} ds$, though the corresponding one for $X$ is computable. Accordingly, we have relied on propagation of uncertainty methods to obtain a closed form solution for the probability density of $Y$, though not explicitly its moments as the logistic normal integral does not lend itself to simple expansions \cite{pirjol2013logistic}.

\subsubsection*{Time slice distributions for $X$ and $Y$}
The time slice distribution is the probability density function of $Y$ from time $t$, that is the one-period representation, starting at $t$ with $y_0=\frac{1}{2}+\frac{1}{2}\text{erf}(x_0)$.  Inversely, for $X$ given $y_0$, the corresponding $x_0$, $X$ may be found to be normally distributed for the period $T-t_0$ with 

$$
\mathbb{E}(X,T)=X_0 e^{ \sigma ^2 (T-t_0)},
$$
$$\mathbb{V}(X,T)=\frac{e^{2  \sigma ^2 (T-t_0)}-1}{2 }$$
and a kurtosis of $3$.
By probability transformation we obtain  $\varphi$, the corresponding distribution of $Y$ with initial value $y_0$ is given by

\begin{dmath}
\varphi(y;y_0,T)=\frac{1}{\sqrt{e^{2 \sigma ^2 (t-t_0)}-1}}\exp \left\{\text{erf}^{-1}(2 y-1)^2-\frac{1}{2} \left(\coth \left(\sigma ^2 t\right)-1\right) \left(\text{erf}^{-1}(2
   y-1)-\text{erf}^{-1}(2 y_0-1) e^{\sigma ^2 (t-t_0)}\right)^2\right\}
\end{dmath}
and we have 
$\mathbb{E}(Y_t)=Y_0$.

As to the variance, $\mathbb{E}(Y^2)$, as mentioned above, does not lend itself to a closed-form solution derived from $\varphi(.)$, nor from the stochastic integral; but it can be easily estimated from the closed form distribution of $X$ using methods of propagation of uncertainty for the first two moments (the delta method).

Since the variance of a function $f$ of a finite moment random variable $X$  can be approximated as  $V\left(f(X)\right)= f'\left(\mathbb{E}(X)\right)^2 V(X)$:

$$\left. \frac{\partial S^{-1}(y)}{\partial y}\right|_{y=Y_0} s^2 \approx \frac{e^{2 \sigma^2 (T-t_0)}-1}{2 }$$

\begin{equation}
s \approx \sqrt{\frac{e^{-2 \text{erf}^{-1}(2 Y_0-1)^2} \left(e^{2 \sigma ^2 (T-t_0)}-1\right)}{2 \pi }}.
\end{equation}
Likewise for calculations in the opposite direction, we find
\begin{equation*}
\sigma \approx \frac{\sqrt{\log \left(2 \pi   s^2 e^{2 \text{erf}^{-1}(2 Y_0-1)^2}+1\right)}}{\sqrt{2} \sqrt{T-t_0}},	
\end{equation*}
which is  \eqref{eqvar} in the presentation of the main result.

Note that expansions including higher moments do not bring a material increase in precision -- although $s$ is highly nonlinear around the center, the range of values for the volatility of the total or, say, the electoral college is too low to affect  higher order terms in a significant way, in addition to the boundedness of the sigmoid-style transformations.

\begin{figure}
	\includegraphics[width=.99\linewidth]{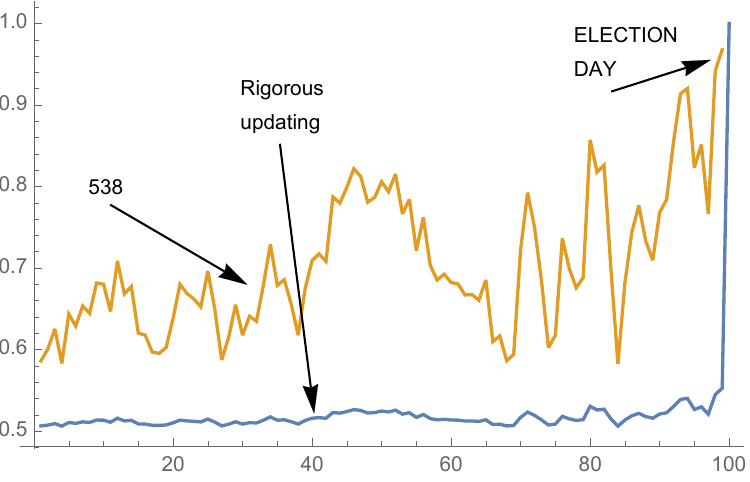}

\caption{Shows the estimation process cannot be in sync with the volatility of the estimation of (electoral or other) votes as it violates arbitrage boundaries}
\end{figure}

\subsection{ A Discussion on Risk Neutrality}
We apply risk neutral valuation, for lack of conviction regarding another way, as a default option. Although $Y$ may not necessarily be tradable, adding a risk premium for the process involved in determining the arbitrage valuation would necessarily imply a negative one for the other candidate(s), which is hard to justify. Further, option values or binary bets, need to satisfy a no Dutch Book argument (the De Finetti form of no-arbitrage) (see \cite{freedman2003notes}), i.e. properly priced binary options interpreted as probability forecasts give no betting "edge" in all outcomes without loss. Finally, any departure from risk neutrality would degrade the Brier Score (about which, below) as it would represent a diversion from the final forecast.

Also note the absence of the assumptions of financing rate usually present in financial discussions.

\section{The Bachelier-Style valuation}

Let $F(.)$ be a function of a variable $X$ satisfying 
\begin{equation}
dX_t= \sigma^2\, X_t  dt+\sigma\, dW_t.
\label{eq:ABM}
\end{equation}

We wish to show that $X$ has a simple Bachelier option price $B(.)$.
The idea of no arbitrage is that a continuously made
forecast must itself be a martingale. 

%

%
%
Applying It\^o's Lemma to $F \triangleq B$ for $X$ satisfying \eqref{eq:ABM}  yields

$$
dF = \left[    \sigma^2\,X\,\frac{\p F}{\p X} +   \frac12 \,\sigma^2\,\frac{\p^2 F}{\p X^2} +   \frac{\p F}{\p t}
\right]\,dt + \sigma\,\frac{\p F}{\p X}\,dW
$$
so that, since $\frac{\p F}{\p t} \triangleq 0$, $F$ must satisfy the partial differential equation
\begin{equation}
  \frac12 \,\sigma^2\,\frac{\p^2 F}{\p X^2} + \sigma^2\,X\,\frac{\p F}{\p X}+\frac{\p F}{\p t} =0,
\end{equation}
which is the driftless condition that makes $B$ a martingale.


For a binary (call) option, we have for terminal conditions $B(X,t) \triangleq F,F_{T}=\theta (x-l)$, where $\theta (.)$ is the Heaviside theta function and $l$ is the threshold:
$$\theta(x) := \begin{cases} 1, & x \ge l \\ 0, & x<l \end{cases}$$
with initial condition $x_0$ at time $t_0$ and terminal condition at $T$ given by:
\begin{equation*}
\frac{1}{2} \text{erfc}\left(\frac{x_0 e^{\sigma ^2 t}-l}{\sqrt{e^{2 \sigma ^2 t}-1}}\right)
\end{equation*}
which is, simply, the survival function of the Normal distribution parametrized under the process for $X$. 

Likewise we note from the earlier argument of one-to one (one can use Borel set arguments ) that $$\theta(y) := \begin{cases} 1, & y \ge S(l)\\ 0, & y<S(l), \end{cases}$$
so we can price the alternative process $B(Y,t)=\mathbb{P}(Y>\frac{1}{2})$ (or any other similarly obtained threshold $l$, by pricing $$B(Y_0,t_0)= \mathbb{P}(x>S^{-1}(l)).$$ 

The pricing from the proportion of votes is given by:
\begin{equation*}
B(Y_0,\sigma,t_0,T)=\frac{1}{2} \text{erfc}\left(\frac{l-\text{erf}^{-1}(2 Y_0-1) e^{\sigma ^2 (T-t_0)}}{\sqrt{e^{2 \sigma ^2 (T-t_0)}-1}}\right),
\end{equation*}
the main equation \eqref{maineq}, which can also be expressed less conveniently as

\begin{dmath*}
B(y_0,\sigma,t_0,T)=\frac{1}{\sqrt{e^{2 \sigma ^2 t}-1}}\int_l^1 \exp \left(\text{erf}^{-1}(2 y-1)^2-\frac{1}{2} \left(\coth \left(\sigma ^2 t\right)-1\right) \left(\text{erf}^{-1}(2 y-1)-\text{erf}^{-1}(2 y_0-1) e^{\text{$\sigma $}^2 t}\right)^2\right) \, dy
\end{dmath*}


\section{Bounded Dual Martingale Process}

\begin{figure}	
	\includegraphics[width=0.8 \linewidth]{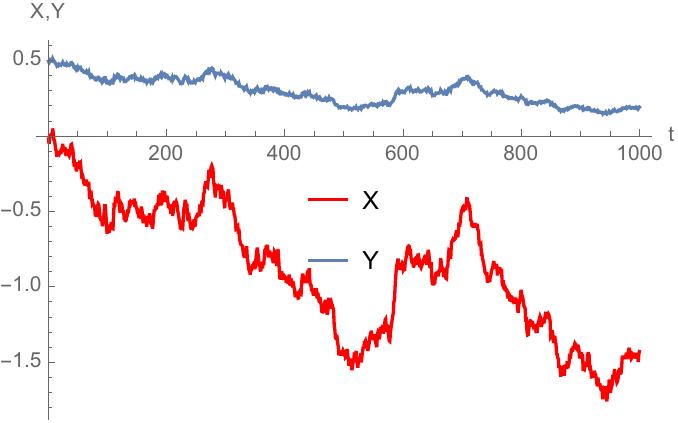}
	\caption{Process and Dual Process }\label{dual}
\end{figure}

$Y_T$ is the terminal value of a process on election day. It lives in $[0,1]$ but can be generalized to the broader $[L,H]$, $L,H \in [0,\infty)$. The threshold for a given candidate to win is fixed at $l$. $Y$ can correspond to raw votes, electoral votes, or any other metric. We assume that $Y_t$ is an intermediate realization of the process at $t$, either produced synthetically from polls (corrected estimates) or other such systems.

Next, we  create, for an unbounded arithmetic stochastic process, a bounded "dual" stochastic process using a sigmoidal transformation. It can be helpful to map processes such as a bounded electoral process to a Brownian motion, or to map a bounded payoff to an unbounded one, see Figure \ref{scheme}.

\begin{proposition}
Under sigmoidal style transformations $S:x \mapsto y, \mathbb{R}\rightarrow [0,1]$ of the form 
a)
$\frac{1}{2} +\frac{1}{2}\text{erf}(x)$,
or 
b) 
$\frac{1}{1+\exp \left(-x\right)}$,
 if $X$ is a martingale, $Y$ is only a martingale for $Y_0=\frac{1}{2}$, and if $Y$ is a martingale, $X$ is only a martingale for $X_0=0$	. 
\end{proposition}

\begin{proof}

The proof is sketched as follows. From It\^o's lemma, the drift term for  $dX_t$ becomes 1) $ \sigma^2 X(t) $,  or 2) $\frac{1}{2}\sigma ^2\text{Tanh}\left(\frac{X(t)}{2}\right)$, where $\sigma$ denotes the volatility, respectively with transformations of the forms a) of $X_t$ and b) of $X_t$ under a martingale for $Y$.  The drift for $dY_t$  becomes: 1) 
$
 \frac{ \sigma ^2 e^{-\text{erf}^{-1}(2 Y-1)^2} \text{erf}^{-1}(2 Y-1)}{\sqrt{\pi }}
 $
 or 
 2)
 $\frac{1}{2} \sigma ^2 Y  (Y-1) (2 Y-1) $ 
  under a martingale for $X$.
\end{proof}

We therefore select the case of $Y$ being a martingale and present the details of the transformation a). The properties of the process have been developed by Carr \cite{carr2017bounded}. 
Let $X$ be the arithmetic Brownian motion \eqref{eq:ABM}, with $X$-dependent drift and constant scale $\sigma$:
\begin{equation*}
dX_t= \sigma^2 X_t  dt+\sigma dW_t \label{ABM},~~ 0<t<T<+\infty.
\end{equation*}

We note that this has similarities with the Ornstein-Uhlenbeck process normally written $dX_t=\theta  (\mu - X_t) dt + \sigma dW$, except that we have $\mu=0$ and violate the rules by using a negative mean reversion coefficient, rather more adequately described as "mean repelling", $\theta=-\sigma^2$.

We map from $X \in (-\infty, \infty)$ to its dual process  $Y$ as follows.
With $S: \mathbb{R} \rightarrow [0,1]$,	$Y= S(x)$,
\begin{equation*}
	S(x)= \frac{1}{2} + \frac{1}{2}\text{erf}(x)
\end{equation*}
the dual process (by unique transformation since 
$S$ is one to one, becomes, for $ y\triangleq S(x)$, using Ito's lemma (since $S(.)$ is twice differentiable and $\p S/ \p t = 0$):
$$dS=\left(\frac{1}{2} \sigma ^2 \frac{\partial^2 S}{\partial x^2}+ X \sigma^2\frac{\partial S}{\partial x}\right) \mathrm{d}t+\sigma  \frac{\partial S}{\partial x} dW   $$
which with zero drift can be written as a process 
\begin{equation*}
dY_t=s(Y) dW_t, 
\end{equation*}
 for all $t>\tau, \mathbb {E}(Y_t|Y_\tau)=Y_\tau.$ and scale
$$s(Y)=\frac{\sigma  }{\sqrt{\pi }}e^{-\text{erf}^{-1}(2 y-1)^2} $$
which as we can see in Figure \ref{instantvol}, $s(y)$ can be approximated by the quadratic function $y (1-y)$ times a constant.
     
 \begin{figure}[h!]
	\includegraphics[width=0.8 \linewidth]{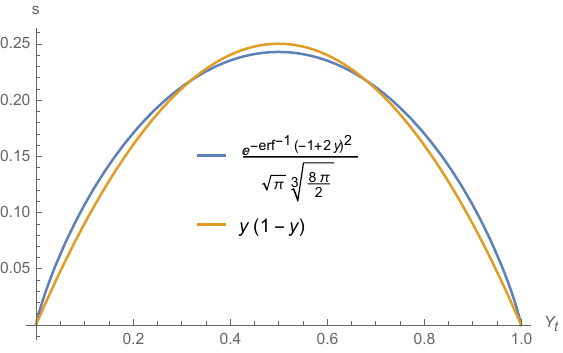}
	\caption{The instantaneous volatility of $Y$ as a function of the level of $Y$ for two different methods of transformations of $X$, which appear to not be substantially different. We compare to the quadratic form $y-y^2$ scaled by a constant $\frac{1}{\sqrt[3]{\frac{8 \pi }{2}}}$. The volatility declines as we move away from $\frac{1}{2}$ and collapses at the edges, thus maintaining $Y$ in $(0,1)$. For simplicity we assumed $\sigma=t=1$. }\label{instantvol}
\end{figure}

   We can recover equation \eqref{eq:ABM} by inverting, namely $S^{-1}(y)=\text{erf}^{-1}(2 y-1)$, and again applying It\^o's Lemma. As a consequence of gauge invariance option prices are identical whether priced on $X$ or $Y$, even if one process has a drift while the other is a martingale. In other words, one may apply one's estimation to the electoral threshold, or to the more complicated $X$ with the same results. And, to summarize our method, pricing an option on $X$ is familiar, as it is exactly a Bachelier-style option price.

\section{Relation to De Finetti's Probability Assessor}

This section provides a brief background for the conventional approach to probability assessment.
De Finetti \cite{de2008philosophical} has shown that the "assessment" of the "probability" of the realization of a random variable in $  \{0,1\}$ requires a nonlinear  loss function -- which makes his definition of \textit{probabilistic assessment} differ from that of the P/L of a trader engaging in binary bets. 

Assume that a betting agent in an $n$-repeated two period model, $t_0$ and $t_1$, produces a strategy $\mathfrak{S}$ of bets $b_{0,i} \in [0,1]$ indexed by $i={1,2,\ldots,n}$, with the realization of the binary r.v. $\mathbbm{1}_{t_1,i}$. If we take the absolute variation of his P/L over $n$ bets, it will be 
$$L_1(\mathfrak{S})=\frac{1}{n}\sum_{i=1}^n\left|\mathbbm{1}_{t_1,i}-b_{t_0,i}\right|.$$ 

For example, assume that $\mathbb{E}(\mathbbm{1}_{t_1})=\frac{1}{2}$. Betting on the probability, here $\frac{1}{2}$, produces a loss of $\frac{1}{2}$ in expectation, which is the same as betting either $0$ or $1$ -- hence not favoring the agent to bet on the exact probability.

If we work with the same random variable and non-time-varying probabilities, the $L^1$ metric would be appropriate:

$$L_1(\mathfrak{S})=\frac{1}{n}\left|\mathbbm{1}_{t_1,i}-\sum_{i=1}^n b_{t_0,i}\right|.$$

De Finetti proposed a "Brier score" type function, a quadratic loss function in $L^2$: $$L_2(\mathfrak{S})=\frac{1}{n}\sum_{i=1}^n(\mathbbm{1}_{t_1,i}-b_{t_0,i})^2,$$ the minimum of which is reached for $b_{t_0,i}= \mathbb{E}(\mathbbm{1}_{t_1})$.

In our world of continuous time derivative valuation, where, in place of a two period lattice model,  we are interested, for the same final outcome at $t_1$, in the stochastic process $b_t$, $t_0\geq t \geq t_1$, the arbitrage "value" of a bet on a binary outcome needs to match the expectation, hence, again, we map to the Brier score -- \textit{by an arbitrage argument}. Although there is no quadratic loss function involved, the fact that the bet is a function of a  martingale, which is required to be itself a  martingale, i.e. that the conditional expectation remains invariant to time, does not allow an arbitrage to take place. A "high" price can be "shorted" by the arbitrageur, a "low" price can be "bought", and so on repeatedly. The consistency between bets at period $t$ and other periods $t+\Delta t$ enforces the probabilistic discipline. In other words, someone can "buy" from the forecaster then "sell" back to him, generating a positive expected "return" if the forecaster is out of line with martingale valuation.

 As to the current practice by forecasters, although some election forecasters appear to be aware of the need to minimize their Brier Score, the idea that the revisions of estimates should also be subjected to martingale valuation is not well established.

\section{Conclusion and Comments}

As can be seen in Figure \ref{graphmain}, a binary option reveals more about uncertainty than about the true estimation, a result well known to traders, see \cite{taleb1997dynamic}.

In the presence of more than $2$ candidates, the process can be generalized with the following heuristic approximation. Establish the stochastic process for $Y_{1,t}$, and just as $Y_{1 ,t}$ is a process in $[0,1]$, $Y_{2,t}$ is a  process  $\in (Y_{1,t},1]$, with $Y_{3,t}$ the residual $1-Y_{2,t}-Y_{1,t}$, and more generally $Y_{n-1,t}\in (Y_{n_2,t},1]$ and $Y_{n,t}$ is the residual $Y_n=1-\sum_{i=1}^{n-1}Y_{i,t}$. For $n $ candidates, the $\text{n}^{th}$ is the residual.

\section{Acknowledgements}

The author thanks Dhruv Madeka and Raphael Douady for
detailed and extensive discussions of the paper as well as thorough
auditing of the proofs across the various iterations, and,
worse, the numerous changes of notation. Peter Carr helped
with discussions on the properties of a bounded martingale and
the transformations. I thank David Shimko, Andrew Lesniewski,
and Andrew Papanicolaou for comments. I thank
Arthur Breitman for guidance with the literature for numerical
approximations of the various logistic-normal integrals. I
thank participants of the Tandon School of Engineering and
Bloomberg Quantitative Finance Seminars. I also thank Bruno
Dupire, MikeLawler, the Editors-In-Chief, and various friendly
people on social media. DhruvMadeka from Bloomberg, while
working on a similar problem, independently came up with the
same relationships between the volatility of an estimate and its
bounds and the same arbitrage bounds. All errors are mine.

\end{document}